\documentclass[acmsmall,screen,nonacm]{acmart}

\AtBeginDocument{%
  }
\usepackage{subcaption}
\usepackage{qcircuit}
\usepackage{wrapfig}
\usepackage{xcolor,listings}
\usepackage{algorithm,algpseudocode}
\usepackage{multirow}
\usepackage{thmtools}
\usepackage{tikz}
\usetikzlibrary{positioning, arrows.meta, fit, backgrounds, calc}

\usepackage{graphicx,color,xcolor,url,setspace}

\def\({\begin{eqnarray*}}
\def\){\end{eqnarray*}}

\def\phi{\varphi}

\def\ket#1{|#1\rangle}
\def\bra#1{\langle#1|}

\def\outer#1#2{\ket{#1}\bra{#2}}

\def\C#1{\text{\sf C}{(#1)}}

\def\hs{\mathcal{D}}
\def\oper{\mathcal{O}}
\def\frob{\mathcal{F}}
\def\trace{\mathcal{T}}

\title{Distance-Independent Universality of Clifford+T}
\author{Jens Palsberg}
\email{palsberg@ucla.edu}
\author{Keli Huang}
\email{kelihuang@cs.ucla.edu}
\author{Abdullah Almanei}
\email{abdullahalmanei@g.ucla.edu}

\affiliation{%
  \institution{UCLA}
  \city{University of California, Los Angeles}
  \country{USA}
}
\date{July, 2026}

\begin{abstract}
A well-known theorem states that the Clifford+T gate set is universal for quantum computing.
The theorem combines approximation using a distance measure on unitary matrices with equality up to global phase.
Previous proofs combine these two notions for specific distance measures, but do not identify the properties that govern how they work together.
Which properties are sufficient to state and prove the theorem?

We answer this question by defining projective distance measures using four axioms. We prove that the universality theorem holds for every distance measure satisfying these axioms. Thus, our formulation of the theorem is independent of any particular choice of distance measure. We show that the Hilbert--Schmidt distance is already a projective distance measure. We also develop a general construction that transforms a large class of distance measures into projective distance measures and apply it to obtain projective versions of the operator-norm distance, the Frobenius distance, and the trace distance. Finally, we formalize the proof of the universality theorem in Lean.
\end{abstract}

\begin{document}
\maketitle


\section{Introduction}
\label{sec:introduction}

A well-known theorem states that the Clifford+T gate set is universal for quantum computing.  In more detail, the theorem says that for any unitary matrix, we can use the Clifford+T gate set to build a quantum circuit that is as close as we like to the unitary matrix.  This requires us to use a distance measure, but the literature offers no canonical choice, as we discuss next.

The well-known book by \citeN{NielsenChuang00} gave a proof of the universality of Clifford+T using the operator-norm distance. 
The idea is to decompose any unitary in two steps.
The first step is to decompose the unitary into a product of 2-qubit gates and 1-qubit gates, up to a global phase.
The second step is to approximate each of those smaller gates by circuits of Clifford+T gates.
We can view the second step as solving a synthesis problem: given a small unitary, find a circuit that is a small distance away and has good properties.
The synthesis problem is important for many applications and has a rich literature.
For example, \citeN{RossSelinger16} gave an algorithm for approximating the 1-qubit $R_z(\theta)$ gates, using the operator-norm distance and producing circuits with few T-gates.
For another example, \citeN{quest-asplos22} gave an algorithm for approximating small unitaries, using the Hilbert-Schmidt distance and producing circuits with good properties such as low depth.
For a third example, \citeN{lbnl25} gave an algorithm for approximating small unitaries, using the operator-norm distance, the Frobenius distance, and the trace-norm distance.
In summary, the community uses a variety of distance measures to support approximation.  This raises the question of which of those distance measures, along with others, are sufficient to state and prove universality of Clifford+T.  Any such distance measure must combine approximation with equality up to a global phase. We crystallize these considerations into the following question about the distance measure.

\begin{quotation}
\noindent
{\em Which properties are sufficient to state and prove universality of Clifford+T?}
\end{quotation}

We answer this question by defining projective distance measures using four axioms. We prove that the universality theorem holds for every distance measure satisfying these axioms. Thus, our formulation of the theorem is independent of any particular choice of distance measure. We show that the Hilbert--Schmidt distance is already a projective distance measure. We also develop a general construction that transforms a large class of distance measures into projective distance measures and apply it to obtain projective versions of the operator-norm distance, the Frobenius distance, and the trace distance. 
We have formalized our proof of the universality theorem in Lean.

\paragraph{The Rest of the Paper}
In Section~\ref{sec:quantum-concepts} we introduce the quantum concepts that we will use in the paper.
In Section~\ref{sec:distance-measures} we give examples and properties of distance measures, and we present our four axioms.
In Section~\ref{sec:the-universality-theorem} we state the universality theorem and two key lemmas, and we give a high-level proof.
In Sections~\ref{sec:clifford-plus-rz-is-universal}--\ref{sec:rz-can-be-approximated}, we prove those two lemmas, with the help of well-known lemmas listed in Appendix~\ref{sec:decomposition-lemmas}.
First, in Section~\ref{sec:clifford-plus-rz-is-universal}, we give a variant of a known proof \cite{shende2005synthesis} that Clifford+$R_z(\theta)$ is universal , 
and then in Section~\ref{sec:rz-can-be-approximated}, we give a variant of a known proof \cite{boykin1999universal, RossSelinger16} that $R_z(\theta)$ can be approximated.
We highlight every place in the proofs where we use the four axioms that define projective distance measures.

\section{Quantum Concepts}
\label{sec:quantum-concepts}

We recall the Pauli matrices $X,Y,Z$, two matrices $S,T$ that satisfy $S^2 = Z$ and $T^2 = S$, the identity matrix $I$, the Hadamard matrix $H$, the $\C{X}$ matrix, and three rotation matrices $R_x(\theta), R_y(\theta), R_z(\theta)$:
\[
\begin{array}{rclrclrclrcl}
X &=&
\left(
\begin{array}{cc}
 0 & 1 \\
 1 & 0
\end{array}
\right)
&
Y &=&
\left(
\begin{array}{cc}
 0 & -i \\
 i & 0
\end{array}
\right)
&
Z &=&
\left(
\begin{array}{cc}
 1 & 0 \\
 0 & -1
\end{array}
\right)
\\ \\
S &=&
\left(
\begin{array}{cc}
 1 & 0 \\
 0 & i
\end{array}
\right)
&
T &=&
\left(
\begin{array}{cc}
 1 & 0 \\
 0 & e^{i\pi/4}
\end{array}
\right)
&
I &=&
\left(
\begin{array}{cc}
 1 & 0 \\
 0 & 1
\end{array}
\right)
&
H &=&
\frac{1}{\sqrt{2}}
\left(
\begin{array}{cc}
 1 & 1 \\
 1 & -1
\end{array}
\right)
\end{array}
\]
\(
\C{X} &=&
\left(
\begin{array}{cccc}
 1 & 0 & 0 & 0 \\
 0 & 1 & 0 & 0 \\
 0 & 0 & 0 & 1 \\
 0 & 0 & 1 & 0 \\
\end{array}
\right)
\)
\[
\begin{array}{rclclcl}
R_x(\theta) &=& 
e^{-i\theta X/2}
&=&
\cos(\theta/2)\ I - i\sin(\theta/2)\ X
&=&
\left(
\begin{array}{cc}
\cos(\theta/2)    & -i \sin(\theta/2) \\
-i \sin(\theta/2) & \cos(\theta/2)
\end{array}
\right)
\\
R_y(\theta) &=&
e^{-i\theta Y/2}
&=&
\cos(\theta/2)\ I - i\sin(\theta/2)\ Y
&=&
\left(
\begin{array}{cc}
\cos(\theta/2)    & -\sin(\theta/2) \\
\sin(\theta/2) & \cos(\theta/2)
\end{array}
\right)
\\
R_z(\theta) &=&
e^{-i\theta Z/2}
&=&
\cos(\theta/2)\ I - i\sin(\theta/2)\ Z
&=&
\left(
\begin{array}{cc}
e^{-i\theta/2} & 0 \\
0              & e^{i\theta/2}
\end{array}
\right)
\end{array}
\]
Those matrices satisfy the following identities:
\begin{eqnarray}
R_x(\theta) &=& H\ R_z(\theta)\ H
\label{eq:rx-in-terms-of-rz}
\\
R_y(\theta) &=& S\ H\ R_z(\theta)\ H\ S^\dagger
\label{eq:ry-in-terms-of-rz}
\end{eqnarray}
For a matrix $U$, we use $U^\dag$ to denote the conjugate transpose of $U$.
A $2^n \times 2^n$ matrix $U$ is unitary if $U U^\dag = U^\dag U = I^{\otimes n}$.
All the above matrices are unitary, and they are also called gates.
We use $U(2^n)$ to denote the set of unitaries of size $2^n \times 2^n$, and we call them $n$-qubit unitaries.
We say that two unitaries $U$ and $V$ are equal up to a global phase,
written $U \sim V$, if there exists a real number $\alpha$ such that
$U = e^{i\alpha} V$, where $e^{i\alpha}$ is called the global phase.

We use the term circuit to mean a product of unitaries. 
The gates $\{\C{X},X,Y,Z,I,H,S\}$ are called {\em Clifford gates} and
can be generated from the smaller set $\{\C{X}, H, S\}$.
In contrast, $T$ is not a Clifford gate, and, in general, 
$R_x(\theta), R_y(\theta), R_z(\theta)$ are not Clifford gates, either.

A square matrix is Hermitian if $A = A^\dag$.  Among the above matrices $X, Y, Z, I, H, \C{X}$ are Hermitian.
For a unitary $U = e^{i\alpha A}$, where 
$\alpha$ is a real number and
$A$ is a Hermitian matrix,
we define $U^a = e^{i a\alpha A}$.
We will apply this definition of exponentiation only to matrices that are written $e^{i \alpha A}$.

In circuit diagrams, we draw the $X$ gate as $\oplus$, a single control as $\bullet$, and a bundle of qubits on a single wire as \,/\,.
Following \cite{ShendeBullockMarkov06}, we use the notation $\Box$ to indicate that a unitary is diagonal on the corresponding qubit and commutes with $Z$ on that qubit.

\section{Distance Measures}
\label{sec:distance-measures}

Here are definitions of the four distance measures mentioned in Section~\ref{sec:introduction}, for $U,V \in U(2^n)$:
\[
\begin{array}{rrclrl}
\text{Hilbert-Schmidt distance:} &
\hs(U,V) &=& 
\sqrt{ 1 - \left(\frac{ | \operatorname{Tr}( U^\dagger V ) | }{ 2^n }\right)^2 }
\label{eq:definition-of-hilbert-schmidt}
\\[0.2cm]
\text{Trace distance:} &
\trace(U,V)
&=&
\|U-V\|_{\mathrm 1} &
\text{where} & \|A\|_{\mathrm 1} = \operatorname{Tr}(\sqrt{A^\dagger A})
\\[0.2cm]
\text{Operator-norm distance:} &
\oper(U,V) &=&
\|U-V\|_{2} &
\text{where} & \|A\|_{2} = \max_{\|x\|=1}\|Ax\|
\\[0.2cm]
\text{Frobenius distance:} &
\frob(U,V) &=& \|U-V\|_{\mathrm F} &
\text{where} &
\|A\|_{\mathrm F}
=
\sqrt{\operatorname{Tr}(A^\dagger A)}
\end{array}
\]
Those distance measures share some properties, four of which we consolidate into the following notion of a traditional distance measure.

\begin{definition}
A \emph{traditional distance measure} is a function
\[
d : U(2^n)\times U(2^n)\to\mathbb R_{\geq 0}
\]
such that:
\begin{description}
\item[Reflexivity:]
$d(U,U) = 0$.

\item[Common-Phase Invariance:]
For every $\theta\in\mathbb R:\ d(e^{i\theta}U,e^{i\theta}V)=d(U,V)$.

\item[Subadditivity:]
$d(U_1U_2,V_1V_2)
\leq
d(U_1,V_1)+d(U_2,V_2)$.

\item[Continuity:]
For every $\epsilon>0$, every real number $t_0$, and every
Hermitian matrix $A$, there exists $\delta>0$ such that, for every
real number $t$, if $|t_0 - t| < \delta$, 
then $d(e^{i t_0 \pi A},\ e^{i t \pi A}) < \epsilon$.
\end{description}
\end{definition}

Intuitively, {\bf Reflexivity} says that an operation has zero distance from itself, while {\bf Common-Phase Invariance} says that changing the global phase in the same way does not change the distance, and {\bf Subadditivity} says that distances add up when circuits are multiplied.
Finally, {\bf Continuity} says that the function that maps $t$ to $e^{i t \pi A}$ is continuous when we use $d$ to measure the distance between two unitaries.

\begin{lemma}
\label{lem:DTOS-are-traditional-distance-measures}
$\hs,\trace,\oper,\frob$ are traditional distance measures.
\end{lemma}

\begin{proof}
See Appendix~\ref{app:proof-that-DTOF-are-traditional-distance-measures}.
\end{proof}

Traditional distance measures are insufficient for stating and proving that Clifford+T is universal.
The missing piece of the puzzle is how a distance measure works for unitaries that are equal up to global phase.
We need a variant of {\bf Common-Phase Invariance}, which we call {\bf Projective Invariance} and which says that any changes to the global phases do not change the distance.
This leads us to define a subset of the traditional distance measures, which we call the projective distance measures.

\begin{definition}
A \emph{projective distance measure} 
is a function $d : U(2^n)\times U(2^n)\to\mathbb R_{\geq 0}$ that satisfies 
{\bf Reflexivity}, {\bf Subadditivity}, {\bf Continuity}, and 
\begin{description}
\item[Projective Invariance:]
If $U_1 \sim U_2$ and $V_1 \sim V_2$, then $d(U_1,V_1) = d(U_2,V_2)$. 
\end{description}
\end{definition}

Notice that every projective distance measure is also a traditional distance measure.
However, among the four distance measures above, only $\hs$ is a projective distance measure.  

\begin{lemma}
\label{lem:D-is-a-projective-distance-measure}
$\hs$ is a projective distance measure.
\end{lemma}

\begin{proof}
Straightforward, using the proof technique in Appendix~\ref{app:proof-that-DTOF-are-traditional-distance-measures}.
\end{proof}

We can observe that $\trace,\oper,\frob$ all fail to satisfy {\bf Projective Invariance}. 
Specifically, for the identity matrix $I$, we have $\trace(I,I) = \oper(I,I) = \frob(I,I) = 0$ and $I\sim(-I)$, but
\[
\begin{array}{rcccccccccl}
\trace(I,-I) &=& \|2I\|_{\mathrm 1} &=& \operatorname{Tr}(\sqrt{(2I)^\dagger (2I)}) &=& 2 \operatorname{Tr}(I) &=& 4 &\neq& 0 
\\[0.1cm]
\oper(I,-I) &=& \|2I\|_{2} &=& \max_{\|x\|=1}\|(2I)x\| &=& \max_{\|x\|=1}\|2x\| &=& 2 &\neq& 0
\\[0.1cm]
\frob(I, -I) &=& \|2I\|_{\mathrm F} &=& \sqrt{\operatorname{Tr}((2I)^\dagger (2I))} &=& 2 \sqrt{\operatorname{Tr}(I)} &=& 2 \sqrt{2} &\neq& 0
\end{array}
\]
The challenge is to formalize what it means to say ``a traditional distance measure plus the idea of equal up to a global phase''.
We will do that based on the following lemma.

\begin{lemma}
\label{lem:projective-minimum-exists}
Let $d$ be a traditional distance measure.  For every pair of unitary
matrices $U,V$, there exists $\theta_0\in\mathbb R$ such that
\[
d(U,e^{i\theta_0}V)
=
\inf_{\theta\in\mathbb R} d(U,e^{i\theta}V).
\]
\end{lemma}

\begin{proof}
See Appendix~\ref{app:proof-that-projective-minimum-exists}.
\end{proof}

By Lemma~\ref{lem:projective-minimum-exists}, the infimum of $d(U,e^{i\theta}V)$ is always equal to a minimum. In other words, the minimum exists. We therefore define the projectivization operator
$\,\widehat{\cdot}\,$ as follows:
\[
\widehat d(U,V)
=
\min_{\theta\in\mathbb R}
d(U,e^{i\theta}V).
\]

\begin{lemma}
\label{lem:from-traditional-distance-measure-to-projective-distance-measure}
If $d$ is a traditional distance measure, then $\widehat d$ is a projective
distance measure.
\end{lemma}

\begin{proof}
See Appendix~\ref{app:proof-of-from-traditional-distance-measure-to-projective-distance-measure}.
\end{proof}

\begin{lemma}
\label{cor:TOS-are-projective-distance-measures}
$\widehat\trace, \widehat\oper, \widehat\frob$ are projective distance measures.
\end{lemma}

\begin{proof}
Combine Lemma~\ref{lem:DTOS-are-traditional-distance-measures} and Lemma~\ref{lem:from-traditional-distance-measure-to-projective-distance-measure}.
\end{proof}


\section{The Universality Theorem}
\label{sec:the-universality-theorem}


For a projective distance measure $d$,
we say that a gate set is universal with respect to $d$ if,
for any unitary $U$ and any $\epsilon > 0$, 
we can use the gate set to define a circuit $C$
such that $d(U,C) < \epsilon$.
Theorem~\ref{thm:clifford-plus-t-is-universal} states that 
$\{\C{X}, H, T\}$ is universal with respect to any projective distance measure.

\begin{theorem}[\bf Clifford+T is universal]
\label{thm:clifford-plus-t-is-universal}
For any projective distance measure $d$,
any unitary $U$, and any $\epsilon > 0$, 
there exists a circuit $C$ built from $\{\C{X}, H, T\}$,
such that $d(U,C) < \epsilon$.
\end{theorem}

Our proof uses a new gate $G_1$, which we define as follows.
\(
G_1 &=& e^{-3i\pi/8}\ THTHT 
\)
The gate $G_1$ enables us to do the proof of Theorem~\ref{thm:clifford-plus-t-is-universal} in three steps:
\[
\text{Any unitary} 
\stackrel{\text{Lemma}~\ref{lem:clifford-plus-rx-is-universal}}{\longrightarrow}
\{\C{X}, H, S, R_z(\theta)\} 
\stackrel{\text{Lemma}~\ref{lem:approximation-of-rz}}{\longrightarrow}
\{\C{X}, H, S, G_1\} 
\stackrel{\text{Immediate}}{\longrightarrow}
\{\C{X}, H, T\}
\]
In the first step, we use Lemma~\ref{lem:clifford-plus-rx-is-universal} to map any unitary $U$ to a circuit that is equal to $U$, up to a global phase, and built from $\{\C{X}, H, S, R_z(\theta)\}$.
Notice that $\C{X}, H, S$ are all Clifford gates, so the remaining challenge lies with the $R_z(\theta)$ gates.
In the second step, we use Lemma~\ref{lem:approximation-of-rz} to map any $R_z(\theta)$ gate to a circuit built from $\{H,S,G_1\}$, in a way that may require approximation.
Finally, in the third step, we use that $G_1$ is equal to a circuit built from $\{H, T\}$, up to a global phase.

\begin{lemma}[\bf Clifford+R$_z(\theta)$ is universal]
For any $n$-qubit unitary $U$,
there exists a circuit $C$ built from
$\{\C{X}, H, S, R_z(\theta)\}$ such that $U \sim C$.
\label{lem:clifford-plus-rx-is-universal}
\end{lemma}

\begin{lemma}[\bf R$_z(\theta)$ can be approximated]
For any projective distance measure $d$,
any real number $\theta$, and $\epsilon > 0$,
there exists a circuit $C$ built from $\{H, S, G_1\}$
such that $d(R_z(\theta),C) < \epsilon$.
\label{lem:approximation-of-rz}
\end{lemma}

\begin{proof}[Proof of Theorem~\ref{thm:clifford-plus-t-is-universal}]
For a unitary $U$ and $\epsilon > 0$,
Lemma~\ref{lem:clifford-plus-rx-is-universal} says that
there exists a circuit $C_1$ built from $\{\C{X}, H, S, R_z(\theta)\}$
such that $U \sim C_1$.

Let $k$ be the number of $R_z$ gates in $C_1$, and
let $\theta_1, \ldots, \theta_k$ be the angles used in those gates.
We have two cases.

In the first case,
if $k=0$, then $C_1$ is built from $\{\C{X},H,S\}$.
Let $C_2$ be a version of $C_1$ in which we have replaced each $S$ by $T^2$.
Notice that $C_2 = C_1$ and $C_2$ is built from $\{\C{X},H,T\}$.
We get: 
\(
d(U,C_2) &=& d(U,C_1) \;\;=\;\; 
d(U,U) \;\;=\;\;
0 \;<\; \epsilon
\)
In the first step, we use $C_2 = C_1$.
In the second step, we use $U \sim C_1$ and the {\bf Projective Invariance} property.
In the third step, we use the {\bf Reflexivity} property.

In the second case, if $k>0$, then
from Lemma~\ref{lem:approximation-of-rz} we get circuits $W_j$ 
built from $\{H,S,G_1\}$ such that
\begin{eqnarray}
d(R_z(\theta_j),W_j) &<& \epsilon/k
\label{eq:approximation-of-r-z}
\end{eqnarray}
Let $C_3$ be a version of $C_1$ in which we have replaced each $R_z(\theta_j)$ by $W_j$.
Notice that $C_3$ is built from $\{\C{X},H,S,G_1\}$.
We calculate:
\begin{equation}
\label{eq:distance-from-u-to-c2}
d(U, C_3) \;\;=\;\; 
d(C_1, C_3) \;\;\leq\;\;
\sum_{j=1}^k d(R_z(\theta_j),W_j) \;\;<\;\;
k \cdot (\epsilon/k) \;\;=\;\; 
\epsilon
\end{equation}
In the first step, we use $U \sim C_1$ and the {\bf Projective Invariance} property.
In the second step, we use the {\bf Subadditivity} property to bound the distance between products by a sum of distances.  Notice that in that sum of distances, some of those distances are between equal unitaries, and in those cases, 
we use the {\bf Reflexivity} property to get that those distances are equal to 0.
In the third step, we use Equation~(\ref{eq:approximation-of-r-z}).

Let $C_4$ be a version of $C_3$ in which 
we have replaced each $S$ by $T^2$, and in which
we have replaced each $G_1$ by $THTHT$.
Notice that $C_4 \sim C_3$ and that
$C_4$ is built from $\{\C{X},H,T\}$.
We calculate:
\(
d(U,C_4) &=& d(U,C_3) \;<\; \epsilon
\)
In the first step, we use the {\bf Projective Invariance} property, and 
in the second step, we use {Equation~(\ref{eq:distance-from-u-to-c2}).}
\end{proof}

Notice that the above proof of Theorem~\ref{thm:clifford-plus-t-is-universal} uses only three of the four properties of the distance measure. We will use the fourth property in the proof of Lemma~\ref{lem:approximation-of-rz}.

\section{Clifford+\texorpdfstring{$R_z(\theta)$}{Rz(theta)} is Universal (Lemma~\ref{lem:clifford-plus-rx-is-universal})}
\label{sec:clifford-plus-rz-is-universal}


First, we show how to implement a multiplexed $R_z(\alpha)$ gate (Lemma~\ref{rzcount}), and then we prove Lemma~\ref{lem:clifford-plus-rx-is-universal} by induction.

\begin{lemma}
Let $R$ be the $n$-qubit gate 
$\sum_{i=0}^{2^{n-1}-1} R_z(\alpha_i) \otimes \outer{i}{i}$.
There exists a circuit $C$ built from $\{\C{X},R_z(\theta)\}$ such that $R = C$.
\label{rzcount}
\end{lemma}

\begin{proof}
We proceed by induction on $n$.

In the base case of $n=1$, we have that $R$ is a single $R_z(\theta)$ gate.

In the induction step, suppose the property holds for $n$, 
where $n \geq 1$.
Consider the case of an $(n+1)$-qubit multiplexed $R_z(\alpha)$ gate $R$.
From Lemma~\ref{rzrz}, we can implement $R$ using two $n$-qubit multiplexed $R_z(\alpha)$ gates and two $\C{X}$ gates. 
Now we apply the induction hypothesis to those two $n$-qubit multiplexed $R_z(\alpha)$ gates and arrive at the desired result.
\end{proof}


Recall that Lemma~\ref{lem:clifford-plus-rx-is-universal} states that
Clifford+R$_z(\theta)$ is universal:
for any $n$-qubit unitary $U$,
there exists a circuit $C$ built from
$\{\C{X}, H, S, R_z(\theta)\}$ such that $U \sim C$.

\begin{proof}[Proof of Lemma~\ref{lem:clifford-plus-rx-is-universal}]
Suppose $U$ is an $n$-qubit unitary.
We proceed by induction on $n$.

In the base case of $n=1$, we first write $U = e^{i\alpha}\ V$, where $\det(V) = 1$, and then based on that and Lemma~\ref{lem:zxz-decomposition} and {Equation~(\ref{eq:rx-in-terms-of-rz})}, we write 
\(
U &\sim& V \;\;=\;\; R_z(\alpha)\ R_x(\beta)\ R_z(\gamma) \;\;=\;\; R_z(\alpha)\ H\ R_z(\beta)\ H\ R_z(\gamma)
\)
The above expression is built from $\{H,R_z(\theta)\}$.

In the induction step, suppose Lemma~\ref{lem:clifford-plus-rx-is-universal} 
holds for $n$, where $n\geq 1$.  
Consider the case of where $U$ is an $(n+1)$-qubit unitary.
We calculate as follows.
\begin{eqnarray}
\Qcircuit @C=0.5em @R=1em {
& \qw & \multigate{2}{U}  & \qw \\
& \qw & \ghost{U} & \qw \\
& / \qw & \ghost{U}  & \qw 
}\nonumber
& \;\;\;\;\;\;\; &
\Qcircuit @C=0.5em @R=1em {
& \qw & \gate{} \qwx[1] & \gate{R}  & \gate{} \qwx[1] & \qw \\
\lstick{= \;\;\;\; } & \qw & \multigate{1}{Q} &  \gate{} \qwx[-1] & \multigate{1}{P} & \qw \\
& / \qw & \ghost{Q}  &  \gate{} \qwx[-1] & \ghost{P}  & \qw 
}\nonumber
\;\;\;\;\;\;\; \;\;\;\;\;\;\;\,
\Qcircuit @C=0.5em @R=1em {
& \qw & \gate{} \qwx[1] & \gate{S^\dag} & \gate{H} & \gate{R_1} & \gate{H} &\gate{S} & \gate{} \qwx[1] & \qw \\
\lstick{= \;\;\;\; } & \qw & \multigate{1}{Q} & \qw & \qw &  \gate{} \qwx[-1] & \qw & \qw & \multigate{1}{P} & \qw \\
& / \qw & \ghost{Q} & \qw & \qw &  \gate{} \qwx[-1] & \qw & \qw & \ghost{P}  & \qw 
}\nonumber
\\[0.4cm] &&
\Qcircuit @C=0.5em @R=1em {
& \qw & \qw & \gate{R_2} \qwx[1] & \qw & \gate{S^\dag} & \gate{H} & \gate{R_1} & \gate{H} &\gate{S} & \qw & \gate{R_0} \qwx[1] & \qw & \qw \\
\lstick{= \;\;\;\; } & \qw & \multigate{1}{N} & \gate{} \qwx[-1] & \multigate{1}{M} & \qw & \qw &  \gate{} \qwx[-1] & \qw & \qw & \multigate{1}{W} & \gate{} \qwx[-1] & \multigate{1}{V} & \qw \\
& / \qw & \ghost{N} & \gate{} \qwx[-1] & \ghost{M} & \qw & \qw &  \gate{} \qwx[-1] & \qw & \qw & \ghost{W} & \gate{} \qwx[-1] & \ghost{V}  & \qw
}\nonumber
\end{eqnarray}

In the first step, according to Lemma~\ref{cosinesine}, 
there exist $2^n$ 1-qubit gates $R_y(\theta_0)$, $R_y(\theta_1)$, $\cdots$, and $R_y(\theta_{2^n-1})$, four $n$-qubit gates $P_0$, $P_1$, $Q_0$, and $Q_1$, and three $(n+1)$-qubit gates $P = \outer{0}{0} \otimes P_{0} + \outer{1}{1} \otimes P_{1}$, $R = \sum^{2^n-1}_{i=0}R_y(\theta_i) \otimes \outer{i}{i}$, and $Q= \outer{0}{0} \otimes Q_{0} + \outer{1}{1} \otimes Q_{1}$, 
such that the first step is valid. 

In the second step, according to {Equation~(\ref{eq:ry-in-terms-of-rz})}, for $R = \sum^{2^n-1}_{i=0}R_y(\theta_i) \otimes \outer{i}{i}$ gate, there exists one $(n+1)$-qubit gate $R_1 = \sum^{2^n-1}_{i=0}R_z(\theta_i) \otimes \outer{i}{i}$, such that 
\begin{eqnarray*}
R & = & \sum^{2^n-1}_{i=0}R_y(\theta_i) \otimes \outer{i}{i} 
\;\;=\;\;
\sum^{2^n-1}_{i=0} ( S \cdot H \cdot R_z(\theta_i) \cdot H \cdot S^\dag )\otimes \outer{i}{i} \\
& = & \sum^{2^n-1}_{i=0} ( S \cdot H )\otimes \outer{i}{i} \cdot  \sum^{2^n-1}_{i=0}   R_z(\theta_i) \otimes \outer{i}{i}  \cdot \sum^{2^n-1}_{i=0} ( H \cdot S^\dag )\otimes \outer{i}{i} \\
& = & (S \otimes I^{\otimes n}) \cdot (H \otimes I^{\otimes n}) \cdot R_1 \cdot (H \otimes I^{\otimes n}) \cdot (S^\dag \otimes I^{\otimes n})
\end{eqnarray*}

In the third step, according to Lemma~\ref{demultiplexing}, for $P = \outer{0}{0} \otimes P_{0} + \outer{1}{1} \otimes P_{1} $ gate, there exist $2^n$ 1-qubit gates $R_z(\alpha_0)$, $R_z(\alpha_1)$, $\cdots$, and $R_z(\alpha_{2^n-1})$, two $n$-qubit gates $V$ and $W$, and one $(n+1)$-qubit gate $R_0 = \sum^{2^n-1}_{i=0}R_{z}(\alpha_i) \otimes \outer{i}{i}$, such that the right-most 
change in the third step is valid. 
Similarly, according to Lemma~\ref{demultiplexing}, for $Q = \outer{0}{0} \otimes Q_{0} + \outer{1}{1} \otimes Q_{1} $ gate, there exist $2^n$ 1-qubit gates $R_z(\beta_0)$, $R_z(\beta_1)$, $\cdots$, and $R_z(\beta_{2^n-1})$, two $n$-qubit gates $M$ and $N$, and one $(n+1)$-qubit gate $R_2 = \sum^{2^n-1}_{i=0}R_{z}(\beta_i) \otimes \outer{i}{i}$, such that left-most change in the third step is valid.

We have arrived at a circuit that uses 
3 multiplexed $(n+1)$-qubit $R_z(\theta)$ gates,
4 $n$-qubit gates,
2 $H$ gates, and
4 $S$ gates
(because the $S^\dag$ gate can be written $S^\dag = SSS$).

We can now apply Lemma~\ref{rzcount} to 
the 3 multiplexed $(n+1)$-qubit $R_z(\theta)$ gates, and
we can apply the induction hypothesis to 
4 $n$-qubit gates.
The result is an implementation of $U$ built from 
$\{\C{X}, H, S, R_z(\theta)\}$.
\end{proof}

An $n$-qubit unitary has $2^n \times 2^n = (2^n)^2 = (2^2)^n = 4^n$ entries.
One can show easily that the above proof of Lemma~\ref{lem:clifford-plus-rx-is-universal} constructs a circuit with $O(4^n)$ gates.

\section{\texorpdfstring{R$_z(\theta)$}{Rz(theta)} Can be Approximated (Lemma~\ref{lem:approximation-of-rz})}
\label{sec:rz-can-be-approximated}

Our goal is to replace each $R_z(\theta)$ gate by a circuit built from $\{H,S,G_1\}$ that is as close to $R_z(\theta)$ as we like.    
Our strategy for doing that relies on an additional gate $G_2$ that we define using Clifford gates and $G_1$:
\(
G_2 &=& 
(H S^2)\ G_1\ (S^2 H)
\)
Our proof of Lemma~\ref{lem:approximation-of-rz} uses $\{G_1,G_2\}$ as a convenient gate set that enables a three-step proof:
\begin{enumerate}
\item We express $R_z(\theta)$ in terms of
real powers of $G_1$ and $G_2$ (Section~\ref{sec:step-1-real-powers-and-no-approximation}).
\item We replace those real powers by positive-integer powers
(Section~\ref{sec:step-2-integer-powers-and-use-of-approximation}).
\item We use those positive-integer powers to complete the proof of Lemma~\ref{lem:approximation-of-rz} (Section~\ref{sec:step-3-proof-of-approximation-of-rz}).
\end{enumerate}
In the first step, we use that  $\{G_1,G_2\}$ supports generalized Euler decomposition (Lemma~\ref{lem:generalized-euler-decomposition}), which enables us to express $R_z(\theta)$ as a product of real powers of $G_1$ and $G_2$.
In the second step, we use that $G_1,G_2$ represent rotations by irrational angles, so positive-integer powers can approximate arbitrary real powers (Lemma~\ref{lem:approximate-g-to-the-m-when-is-a-rotation-by-an-irr-multiple-of-pi}). 
In the third step, we put these ideas together to show that 
the gate set $\{G_1,G_2\}$, hence the gate set $\{H,S,G_1\}$, can approximate any $R_z(\theta)$ gate.

\subsection{Step 1: Real Powers and No Approximation}
\label{sec:step-1-real-powers-and-no-approximation}

We will represent the unitary matrices $G_1, G_2$ as exponentiations of Hermitian matrices.  The constants and Hermitian matrices that we need are the following.
\(
&&
p=\sqrt{\frac{10-4\sqrt2}{17}}
\;\;\;\;\;\;\;\;
q=\sqrt{\frac{7+4\sqrt2}{17}}
\;\;\;\;\;\;\;\;
\;\;\;\;\;\;\;\;
A_1=-pX-qZ
\;\;\;\;\;\;\;\;
A_2=-qX+pZ
\\
&&
\lambda_0
=
\frac{1}{\pi}
\arccos\!\left(
\frac{1}{2}
\sqrt{1 + 1/\sqrt{2}}
\right)
\)
Now we turn to the first step of the proof.

\begin{lemma}
\label{lem:properties-of-a1-a2-lambda}
The matrices $A_1$ and $A_2$ are Hermitian, and
$A_1^2=A_2^2=I$ and
$A_1A_2=-A_2A_1$ and
$G_i=e^{i\lambda_0\pi A_i}$.
Moreover, $\lambda_0$ is irrational.
\end{lemma}

\begin{proof}
The proofs of the matrix identities are by straightforward matrix calculation.

For the proof that $\lambda_0$ is irrational, 
suppose, for contradiction, that $\lambda_0$ is rational.
Then $e^{i\lambda_0\pi}$ is a root of unity, so it is an algebraic integer.
Therefore,
$e^{-i\lambda_0\pi}$ is also an algebraic integer.
Let $x = e^{i\lambda_0\pi} + e^{-i\lambda_0\pi}$.
Since algebraic integers are closed under addition, we get:
\[
x
\;\;=\;\;
e^{i\lambda_0\pi}+e^{-i\lambda_0\pi}
\;\;=\;\;
2\cos(\lambda_0\pi)
\;\;=\;\;
\sqrt{1+1/\sqrt{2}}
\]
is also an algebraic integer.
From the above, we derive that
\(
x^2-1 &=& \frac{1}{\sqrt{2}}
\)
and since algebraic integers are closed under multiplication and subtraction,
we get
that $\frac{1}{\sqrt{2}}$ is an algebraic integer.
However, the primitive integer polynomial for $\frac{1}{\sqrt{2}}$ is
$(2X^2-1)$,
which is not monic, so $\frac{1}{\sqrt{2}}$ is not an algebraic integer.
This is a contradiction so $\lambda_0$ is irrational.
\end{proof}

For the remainder of the proof, 
we will not need the explicit formulas for $p, q, A_1, A_2, \lambda_0$, 
and instead we will rely on the properties that are expressed by
Lemma~\ref{lem:properties-of-a1-a2-lambda}.

We can define $R_z(\theta)$ in terms of real powers of $G_1$ and $G_2$, as expressed by Lemma~\ref{lem:from-rz-to-g1-g2}, with no need for approximation.

\begin{lemma}
\label{lem:from-rz-to-g1-g2}
For any $\theta$, there exist real numbers $a, b, c$ such that
$R_z(\theta) = G_1^a\ G_2^b\ G_1^c$.
\end{lemma}

\begin{proof}
We have that $R_z(\theta)$ is a unitary matrix and $\det(R_z(\theta)) = 1$, and 
from Lemma~\ref{lem:properties-of-a1-a2-lambda}, we have $A_1^2=A_2^2=I$ and $A_1A_2=-A_2A_1$.
From those properties and Lemma~\ref{lem:generalized-euler-decomposition}, there exist real numbers
$x,y,z$ such that
\[
R_z(\theta)
=
e^{ixA_1}e^{iyA_2}e^{izA_1}.
\]
We can write:
\[
G_1^a = e^{ia\lambda_0\pi A_1},
\qquad
G_2^b = e^{ib\lambda_0\pi A_2}.
\]
Thus, choosing
\[
a=\frac{x}{\lambda_0\pi},
\qquad
b=\frac{y}{\lambda_0\pi},
\qquad
c=\frac{z}{\lambda_0\pi},
\]
we obtain
\[
R_z(\theta) 
=
G_1^a\ G_2^b\ G_1^c.
\]
\end{proof}

\subsection{Step 2: Positive-Integers Powers and Use of Approximation}
\label{sec:step-2-integer-powers-and-use-of-approximation}

If in Lemma~\ref{lem:from-rz-to-g1-g2}, the real numbers $a,b,c$ are positive integers, we are done:
$G_1^a\ G_2^b\ G_1^c$ is a circuit built from $\{G_1,G_2\}$, 
up to a global phase.
Otherwise, the second step is to find positive integers $m,n,p$ that we can use 
instead of the real numbers $a,b,c$ in
Lemma~\ref{lem:from-rz-to-g1-g2} and 
get us as close as we like to $R_z(\theta)$.
The idea is that irrational rotations have positive-integer powers that 
are dense among all rotations about the same axis.
The key lemma for this approximation is 
Lemma~\ref{%
lem:approximate-g-to-the-m-when-is-a-rotation-by-an-irr-multiple-of-pi}.
Intuitively, 
Lemma~\ref{%
lem:approximate-g-to-the-m-when-is-a-rotation-by-an-irr-multiple-of-pi}
talks about approximation based on a 1-qubit unitary $G$ that is a rotation by 
an irrational multiple of $\pi$.


\begin{lemma}
\label{lem:approximate-g-to-the-m-when-is-a-rotation-by-an-irr-multiple-of-pi}
For any projective distance measure $d$,
any irrational number $\lambda$,
any $G = e^{i\lambda \pi A}$, where 
$A$ is a 1-qubit Hermitian matrix that satisfies $e^{i 2 \pi A} = I$, 
any $\epsilon > 0$, and any real number $a$, 
there exists an positive integer $m$ such that 
$d(G^a, G^m) < \epsilon$.
\end{lemma}

\begin{proof}
We use the {\bf Continuity} property
to find $\delta > 0$ such that for all real numbers $t$, 
if $|a\lambda - t| < \delta$, 
then $d(e^{i a\lambda \pi A},\ e^{i t \pi A}) < \epsilon$.
We have from Kronecker's theorem
\cite[Theorem 439]{HardyWright08} that $\{\,m\lambda\ (\bmod\ 1) \mid m\in\mathbb N\,\}$ is dense in $[0,1)$, which we can scale by a factor of 2 and get that 
$\{\,m\lambda\ (\bmod\ 2) \mid m\in\mathbb N\,\}$ is dense in $[0,2)$.  Therefore, for our chosen $\delta$, there exist a positive integer $m$ and an integer $k$, such that
\(
|a\lambda - (m \lambda + 2k)| &<& \delta
\)
which gives
\(
d(e^{i (a \lambda) \pi A},\ e^{i (m \lambda + 2k) \pi A})
&<&
\epsilon
\)
Notice that
\(
e^{i (a \lambda) \pi A} 
&=&
G^a
\\
e^{i (m \lambda + 2k) \pi A}
&=&
e^{i (m \lambda) \pi A}
\;\;=\;\;
G^m
\)
so we conclude 
$d(G^a, G^m) < \epsilon$.
\end{proof}

\subsection{Step 3: Proof of Lemma~\ref{lem:approximation-of-rz}}
\label{sec:step-3-proof-of-approximation-of-rz}

The third step is to use
Lemmas~\ref{lem:properties-of-a1-a2-lambda}, \ref{lem:from-rz-to-g1-g2}, and \ref{lem:approximate-g-to-the-m-when-is-a-rotation-by-an-irr-multiple-of-pi} in
the proof Lemma~\ref{lem:approximation-of-rz}.
Recall that Lemma~\ref{lem:approximation-of-rz} states that $R_z(\theta)$ can be approximated:
for any projective distance measure $d$,
any real number $\theta$, and $\epsilon > 0$,
there exists a circuit $C$ built from $\{H, S, G_1\}$,
such that $d(R_z(\theta),C) < \epsilon$.

\begin{proof}[Proof of Lemma~\ref{lem:approximation-of-rz}]
For a given $\theta$, 
we use Lemma~\ref{lem:from-rz-to-g1-g2} to find real numbers $a,b,c$ such that
\begin{eqnarray}
R_z(\theta) &=& G_1^a\ G_2^b\ G_1^c
\label{eq:rz-sim-g1-a-g2-b-g1-c}
\end{eqnarray}
From Lemma~\ref{lem:properties-of-a1-a2-lambda},
we have that 
$A_1^2 = A_2^2 = I$ and
$G_i=e^{i\lambda_0\pi A_i}$ and
$\lambda_0$ is irrational.
From $A_1^2 = A_2^2 = I$, we get
\begin{eqnarray}
e^{i 2 \pi A_1} = I 
&&
e^{i 2 \pi A_2} = I
\label{eq:powers-of-a1-and-a2-are-i}
\end{eqnarray}
For a given $\epsilon > 0$, we use {Equation~(\ref{eq:powers-of-a1-and-a2-are-i})} and Lemma~\ref{lem:approximate-g-to-the-m-when-is-a-rotation-by-an-irr-multiple-of-pi}
to find positive integers $m,n,p$ such that:
\(
d(G_1^a, G_1^m) < \epsilon/3 &\;\;\;\;
d(G_2^b, G_2^n) < \epsilon/3 &\;\;\;\;
d(G_1^c, G_1^p) < \epsilon/3
\)
Now we can calculate as follows:
\begin{eqnarray}
d(R_z(\theta),\ G_1^m\ G_2^n\ G_1^p) 
&=& 
d(G_1^a\ G_2^b\ G_1^c,\ G_1^m\ G_2^n\ G_1^p) 
\nonumber
\\
&\leq&
d(G_1^a, G_1^m) + d(G_2^b, G_2^n) + d(G_1^c, G_1^p)
\nonumber
\\
&<&
\epsilon/3 + \epsilon/3 + \epsilon/3
\;\;=\;\; 
\epsilon
\nonumber
\end{eqnarray}
In the first step, we use {Equation~(\ref{eq:rz-sim-g1-a-g2-b-g1-c}).}
In the second step, we use the {\bf Subadditivity} property.

Notice that 
\(
G_2^n &=&
(H S^2)\ G_1 (S^2 H)\  (H S^2)\ G_1 (S^2 H)\ \ldots\ (H S^2)\ G_1 (S^2 H)\ 
\;\;=\;\;
H S^2\ G_1^n\ S^2 H
\\
G_1^m\ G_2^n\ G_1^p &=&
G_1^m\ H S^2\ G_1^n\ S^2 H\ G_1^p
\)
and that $G_1^m\ H S^2\ G_1^n\ S^2 H\ G_1^p$  is built from $\{H,S,G_1\}$.
\end{proof}


%

\section{Conclusion}

A universal gate set can approximate any unitary.
Previous work has used a variety of distance measures to support approximation.  We have axiomatized the notion of projective distance measure and shown that the Hilbert-Schmidt distance satisfies our axioms. In contrast, the operator-norm distance, the Frobenius distance, and the trace distance need to undergo projectivization before they satisfy our axioms.
We used our notion of projective distance measure to give a distance-independent proof of universality of Clifford+T.
Our proof highlights that the only place where approximation is needed is in the implementation of $R_z(\theta)$ gates.
We have formalized our proof of the universality theorem in Lean: here are links to 
\href{https://github.com/AbdullahAlmanei/clifford-universality-lean}
     {the main page} 
and
\href{https://abdullahalmanei.github.io/clifford-universality-lean}
     {the documentation}, 
which includes the blueprint and the dependency graph.

Author contribution statement: all authors contributed significantly.  Specifically, Jens Palsberg wrote a draft of the paper, Keli Huang wrote the first drafts of Section 5 and Appendix D and contributed numerous technical improvements to the writing of the entire paper, and Abdullah Almanei formalized the proof of the universality theorem in Lean.  We used AI to help convince ourselves that the gate $G_1$ is a good choice for our proof, write drafts of some of the cases of the proofs in the appendices, and help discover aspects of the formalization in Lean.

We are supported by the NSF QLCI program through grant number OMA-2016245 and the NSF NQVL program through grant number 2435382.

\appendix

\section{Proof of Lemma~\ref{lem:DTOS-are-traditional-distance-measures}}
\label{app:proof-that-DTOF-are-traditional-distance-measures}

We begin with a couple of lemmas about $\hs$.

\begin{lemma}
$\hs(U_1 U_2, V_1 V_2) = \hs(V_1^\dag U_1, V_2 U_2^\dag)$.
\label{lem:distance-equality-for-products}
\end{lemma}

\begin{proof}
Notice that
\begin{eqnarray}
Tr((U_1 U_2)^\dag\ (V_1 V_2)) 
\;=\;
Tr(U_2^\dag U_1^\dag V_1 V_2) 
\;=\;
Tr(U_1^\dag V_1 V_2 U_2^\dag)
\;=\;
Tr((V_1^\dag U_1)^\dag V_2 U_2^\dag)
\label{eq:trace-equality-between-products}
\end{eqnarray}
In the first step, we use $(UV)^\dag = V^\dag\ U^\dag$.
In the second step, we use $Tr(UV) = Tr(VU)$.
In the third step, we use $(UV)^\dag = V^\dag\ U^\dag$.

Now we calculate:
\(
\hs(U_1 U_2, V_1 V_2) &=&
\sqrt{ 1 - \left(\frac{ | Tr( (U_1 U_2)^\dag\ (V_1 V_2) ) | }{ 2^n }\right)^2 }
\\ &=&
\sqrt{ 1 - \left(\frac{ | Tr( (V_1^\dag U_1)^\dag\ (V_2 U_2^\dag) ) | }{ 2^n }\right)^2 }
\\ &=&
\hs(V_1^\dag U_1, V_2 U_2^\dag)
\)
In the second step, we use {Equation~(\ref{eq:trace-equality-between-products}).}
\end{proof}

Lemma~\ref{lem:distance-equality-for-products} implies 
the following properties, via picking some of $U_1, U_2, V_1, V_2$ to be $I^{\otimes n}$.
\begin{eqnarray}
\hs(I^{\otimes n}, V_1 V_2) &=& \hs(V_1^\dag, V_2)
\label{lem:distance-from-identity-to-product}
\\
\hs(U_1 U_2, I^{\otimes n}) &=& \hs(U_1, U_2^\dag)
\label{lem:distance-from-product-to-identity}
\\
\hs(U,V) &=& \hs(V^\dag, U^\dag)
\label{lem:distance-of-conjugated-unitaries}
\end{eqnarray}

For approximation of a product of gates, errors add at most linearly.

\begin{lemma}
$\hs(U,V) \leq \hs(U,I^{\otimes n}) + \hs(I^{\otimes n},V)$.
\label{lem:trace-inequality}
\end{lemma}

\begin{proof}
\citeN{WangZhang94} state this property as their {Equation~(1)}, though in a different notation.
\end{proof}

Recall that Lemma~\ref{lem:DTOS-are-traditional-distance-measures} states that
$\hs,\trace,\oper,\frob$ are traditional distance measures.

\begin{proof}[Proof of Lemma~\ref{lem:DTOS-are-traditional-distance-measures}]
We consider $\hs,\trace,\oper,\frob$ in turn.

First, we prove that $\hs$ has the four properties of a traditional distance measure.

\paragraph{Reflexivity.}
We calculate:
\(
\hs(U,U) &=& 
\sqrt{ 1 - \left(\frac{ | Tr( U^\dag U ) | }{ 2^n }\right)^2 }
 \;\;=\;\; 
\sqrt{ 1 - \left(\frac{ 2^n }{ 2^n }\right)^2 } 
 \;\;=\;\; 0
\)

\paragraph{Common-Phase Invariance.}
We calculate:
\(
\hs(e^{i\theta}U,e^{i\theta}V) 
&=&
\sqrt{ 1 - \left(\frac{ | Tr( (e^{i\theta} U)^\dagger (e^{i\theta} V) ) | }{ 2^n }\right)^2 }
\\
&=&
\sqrt{ 1 - \left(\frac{ | e^{-i\theta} e^{i\theta} Tr( U^\dagger V ) | }{ 2^n }\right)^2 }
\\
&=&
\sqrt{ 1 - \left(\frac{ | Tr( U^\dagger V ) | }{ 2^n }\right)^2 }
\\
&=&
\hs(U,V)
\)

\paragraph{Subadditivity.}
We calculate:
\(
\hs(U_1 U_2, V_1 V_2) &=&
\hs(V_1^\dag U_1, V_2 U_2^\dag)
\\ &\leq& 
\hs(V_1^\dag U_1, I^{\otimes n}) + \hs(I^{\otimes n}, V_2 U_2^\dag)
\\ &=&
\hs(V_1^\dag, U_1^\dag) + \hs(V_2^\dag, U_2^\dag)
\\ &=&
\hs(U_1,V_1) + \hs(U_2,V_2)
\)
In the first step, we use Lemma~\ref{lem:distance-equality-for-products}.
In the second step, we use Lemma~\ref{lem:trace-inequality}.
In the third step, we use 
{Equations~(\ref{lem:distance-from-identity-to-product})--(\ref{lem:distance-from-product-to-identity}).}
In the fourth step, we use {Equation~(\ref{lem:distance-of-conjugated-unitaries}).}

\paragraph{Continuity.}
Let $\epsilon>0$, let $t_0\in\mathbb R$, and let $A$ be Hermitian.
For Continuity, we have
\(
\hs(e^{it_0\pi A},e^{it\pi A})
&=&
\sqrt{1-
\left(
\frac{|Tr(e^{i(t-t_0)\pi A})|}{2^n}
\right)^2} 
\)
Notice that 
$\lim_{t\to t_0} |Tr(e^{i(t-t_0)\pi A})| = 2^n$.
Hence,
\[
\lim_{t\to t_0}
\hs(e^{it_0\pi A},e^{it\pi A})
=0,
\]
as required.
Thus $\hs$ satisfies {\bf Continuity}.

In summary, $\hs$ is a traditional distance measure.

\bigskip

Second, we prove that $\trace$ has the four properties of a traditional distance measure.

\paragraph{Reflexivity.}
We have
\[
\trace(U,U) = \|U-U\|_{\mathrm 1} = 0
\]

\paragraph{Common-Phase Invariance.}
For every $\theta\in\mathbb R$,
\[
\begin{aligned}
\trace(e^{i\theta}U,e^{i\theta}V)
&=
\|e^{i\theta}U-e^{i\theta}V\|_{\mathrm 1}
\\
&=
\|e^{i\theta}(U-V)\|_{\mathrm 1}
\\
&=
\|U-V\|_{\mathrm 1}
\\
&=
\trace(U,V).
\end{aligned}
\]

\paragraph{Subadditivity.}
For unitary matrices $U_1,U_2,V_1,V_2$, we calculate:
\[
\begin{aligned}
\trace(U_1U_2,V_1V_2)
&=
\|U_1U_2-V_1V_2\|_{\mathrm 1}
\\
&=
\|(U_1-V_1)U_2+V_1(U_2-V_2)\|_{\mathrm 1}
\\
&\leq
\|(U_1-V_1)U_2\|_{\mathrm 1}
+
\|V_1(U_2-V_2)\|_{\mathrm 1}
\\
&=
\|U_1-V_1\|_{\mathrm 1}
+
\|U_2-V_2\|_{\mathrm 1}
\\
&=
\trace(U_1,V_1)
+
\trace(U_2,V_2).
\end{aligned}
\]
In the third step, we use that the trace norm satisfies the triangle inequality.
In the fourth step, we use the invariance of the trace norm under multiplication by unitary matrices.

\paragraph{Continuity.}
Let $\epsilon>0$, let $t_0\in\mathbb R$, and let $A$ be Hermitian.
We have
\[
\begin{aligned}
\trace(e^{it_0\pi A},e^{it\pi A})
&=
\left\|
e^{it_0\pi A}-e^{it\pi A}
\right\|_{\mathrm 1}
\\
&=
\left\|
I^{\otimes n}-e^{i(t-t_0)\pi A}
\right\|_{\mathrm 1}.
\end{aligned}
\]
Since
\[
\lim_{s\to0}e^{is\pi A}=I^{\otimes n},
\]
there exists $\delta>0$ such that
\[
|t_0-t|<\delta
\quad\Longrightarrow\quad
\left\|
I^{\otimes n}-e^{i(t-t_0)\pi A}
\right\|_{\mathrm 1}
<
\epsilon.
\]
Thus $\trace$ satisfies {\bf Continuity}.

In summary, $\trace$ is a traditional distance measure.

\bigskip

Third, we prove that $\oper$ has the four properties of a traditional distance measure.

We prove the four properties of a traditional distance measure.

\paragraph{Reflexivity.}
We have
\[
\oper(U,U) = \|U-U\|_{\mathrm 2} = 0
\]

\paragraph{Common-Phase Invariance.}
For every $\theta\in\mathbb R$,
\[
\oper(e^{i\theta}U,e^{i\theta}V)
=
\|e^{i\theta}U-e^{i\theta}V\|_{2}
=
\|e^{i\theta}(U-V)\|_{2}
=
\|U-V\|_{2}
=
\oper(U,V).
\]

\paragraph{Subadditivity.}
For unitary matrices $U_1,U_2,V_1,V_2$, we calculate:
\[
\begin{aligned}
\oper(U_1U_2,V_1V_2)
&=
\|U_1U_2-V_1V_2\|_{2}
\\
&=
\|(U_1-V_1)U_2+V_1(U_2-V_2)\|_{2}
\\
&\le
\|(U_1-V_1)U_2\|_{2}
+
\|V_1(U_2-V_2)\|_{2}
\\
&\le
\|U_1-V_1\|_{2}\|U_2\|_{2}
+
\|V_1\|_{2}\|U_2-V_2\|_{2}
\\
&=
\|U_1-V_1\|_{2}
+
\|U_2-V_2\|_{2}
\\
&=
\oper(U_1,V_1)
+
\oper(U_2,V_2).
\end{aligned}
\]
In the third step, we use that the operator norm satisfies
the triangle inequality.
In the fourth step, we use the submultiplicativity
of the operator norm.
In the fifth step, we use that every unitary matrix has
operator norm~$1$.

\paragraph{Continuity.}
Let $\epsilon>0$, let $t_0\in\mathbb R$, and let $A$ be Hermitian.
We have
\[
\begin{aligned}
\oper(e^{it_0\pi A},e^{it\pi A})
&=
\left\|
e^{it_0\pi A}-e^{it\pi A}
\right\|_{2}
\\
&=
\left\|
I^{\otimes n}-e^{i(t-t_0)\pi A}
\right\|_{2}.
\end{aligned}
\]
Since
\[
\lim_{s\to0}e^{is\pi A}=I^{\otimes n},
\]
there exists $\delta>0$ such that
\[
|t_0-t|<\delta
\quad\Longrightarrow\quad
\left\|
I^{\otimes n}-e^{i(t-t_0)\pi A}
\right\|_{2}
<
\epsilon.
\]
Thus $\oper$ satisfies {\bf Continuity}.

In summary, $\oper$ is a traditional distance measure.

\bigskip

Fourth, we prove that $\frob$ has the four properties of a traditional distance measure.

We prove the four properties of a traditional distance measure.

\paragraph{Reflexivity.}
We have
\[
\frob(U,U) = \|U-U\|_{\mathrm F} = 0
\]

\paragraph{Common-Phase Invariance.}
For every $\theta\in\mathbb R$,
\[
\frob(e^{i\theta}U,e^{i\theta}V)
\;=\;
\|e^{i\theta}U-e^{i\theta}V\|_{\mathrm F}
\;=\;
\|e^{i\theta}(U-V)\|_{\mathrm F}
\;=\;
\|U-V\|_{\mathrm F}
\;=\;
\frob(U,V).
\]

\paragraph{Subadditivity.}
For unitary matrices $U_1,U_2,V_1,V_2$, we calculate:
\[
\begin{aligned}
\frob(U_1U_2,V_1V_2)
&=
\|U_1U_2-V_1V_2\|_{\mathrm F}
\\
&=
\|(U_1-V_1)U_2+V_1(U_2-V_2)\|_{\mathrm F}
\\
&\leq
\|(U_1-V_1)U_2\|_{\mathrm F}
+
\|V_1(U_2-V_2)\|_{\mathrm F}
\\
&=
\|U_1-V_1\|_{\mathrm F}
+
\|U_2-V_2\|_{\mathrm F}
\\
&=
\frob(U_1,V_1)+\frob(U_2,V_2).
\end{aligned}
\]
The third step, we use that Frobenius norm satisfies the triangle inequality.
In the fourth step, we use the invariance of the Frobenius norm
under multiplication by unitary matrices.

\paragraph{Continuity.}
Let $\epsilon>0$, let $t_0\in\mathbb R$, and let $A$ be Hermitian.
We have
\[
\begin{aligned}
\frob(e^{it_0\pi A},e^{it\pi A})
&=
\left\|
e^{it_0\pi A}-e^{it\pi A}
\right\|_{\mathrm F}
\\
&=
\left\|
I^{\otimes n}-e^{i(t-t_0)\pi A}
\right\|_{\mathrm F}.
\end{aligned}
\]
Since
\[
\lim_{s\to 0}e^{is\pi A}=I^{\otimes n},
\]
there exists $\delta>0$ such that
\[
|t_0-t|<\delta
\quad\Longrightarrow\quad
\left\|
I^{\otimes n}-e^{i(t-t_0)\pi A}
\right\|_{\mathrm F}<\epsilon.
\]
Thus $\frob$ satisfies {\bf Continuity}.

In summary, $\frob$ is a traditional distance measure.
\end{proof}


\section{Proof of Lemma~\ref{lem:projective-minimum-exists}}
\label{app:proof-that-projective-minimum-exists}

Recall that Lemma~\ref{lem:projective-minimum-exists} states the following.
Let $d$ be a traditional distance measure.  For every pair of unitary
matrices $U,V$, there exists $\theta_0\in\mathbb R$ such that
\[
d(U,e^{i\theta_0}V)
=
\inf_{\theta\in\mathbb R} d(U,e^{i\theta}V).
\]

\begin{proof}
Fix unitary matrices $U,V$, and define
\[
f(\theta)=d(U,e^{i\theta}V).
\]
We first show that $f$ is continuous.

Let $\epsilon>0$.  By the {\bf Continuity} property, applied to the real number $0$ 
and the
Hermitian matrix $I^{\otimes n}$, 
there exists
$\delta'>0$ such that,
for every real number $s$,
\[
|s|<\delta'
\quad\Longrightarrow\quad
d(I^{\otimes n},e^{is\pi}I^{\otimes n})<\epsilon.
\]
where we use that $e^{is\pi I^{\otimes n}} = e^{is\pi}I^{\otimes n}$.
Set $\delta=\pi\delta'$.  If $|\theta-\phi|<\delta$, then
\[
\left|\frac{\theta-\phi}{\pi}\right|<\delta',
\]
and hence
\[
d(I^{\otimes n},e^{i(\theta-\phi)}I^{\otimes n})<\epsilon.
\]
We calculate:
\(
f(\theta)
&=&
d(U,e^{i\theta}V)
\;\;=\;\;
d\bigl(UI^{\otimes n},e^{i\phi}V e^{i(\theta-\phi)}I^{\otimes n}\bigr)
\;\;\leq\;\;
d(U,e^{i\phi}V)
+
d(I^{\otimes n},e^{i(\theta-\phi)}I^{\otimes n}) 
\\
&<&
f(\phi)+\epsilon.
\)
In the third step, we use {\bf Subadditivity}.
We can interchange $\theta$ and $\phi$ in the above argument and get:
\[
f(\phi)<f(\theta)+\epsilon.
\]
Therefore,
\[
|f(\theta)-f(\phi)|<\epsilon,
\]
so $f$ is continuous.

Since $f$ is $2\pi$-periodic,
\[
\inf_{\theta\in\mathbb R} f(\theta)
=
\inf_{\theta\in[0,2\pi]} f(\theta).
\]
The interval $[0,2\pi]$ is compact, and $f$ is continuous.
Therefore, by the Extreme Value Theorem, there exists
$\theta_0\in\mathbb R$ such that
\[
f(\theta_0)
=
\inf_{\theta\in\mathbb R} f(\theta).
\]

Therefore, there exists $\theta_0\in\mathbb R$ such that
\[
d(U,e^{i\theta_0}V)
=
\inf_{\theta\in\mathbb R} d(U,e^{i\theta}V).
\]
\end{proof}

\section{Proof of Lemma~\ref{lem:from-traditional-distance-measure-to-projective-distance-measure}}
\label{app:proof-of-from-traditional-distance-measure-to-projective-distance-measure}

Recall that Lemma~\ref{lem:from-traditional-distance-measure-to-projective-distance-measure} states the following.
If $d$ is a traditional distance measure, then $\widehat d$ is a projective
distance measure.

\begin{proof}
We show that $\widehat d$ satisfies the four properties of a
projective distance measure.

\paragraph{Reflexivity.}
\[
\widehat d(U,U)
=
\min_{\theta\in\mathbb R}
d(U,e^{i\theta}U)
\leq
d(U,U)
=
0.
\]
In the last step, we use the {\bf Reflexivity} property.
Since $\widehat d$ is nonnegative, we conclude that
$\widehat d(U,V)=0$.

\paragraph{Projective Invariance.}
Let $\alpha,\beta\in\mathbb R$.  Then
\[
\begin{aligned}
\widehat d(e^{i\alpha}U,e^{i\beta}V)
&=
\min_{\theta\in\mathbb R}
d(e^{i\alpha}U,e^{i(\theta+\beta)}V)
\\
&=
\min_{\theta\in\mathbb R}
d(e^{-i\alpha}e^{i\alpha}U,
  e^{-i\alpha}e^{i(\theta+\beta)}V)
\\
&=
\min_{\theta\in\mathbb R}
d(U,e^{i(\theta+\beta-\alpha)}V)
\\
&=
\min_{\phi\in\mathbb R}
d(U,e^{i\phi}V)
\\
&=
\widehat d(U,V),
\end{aligned}
\]
where $\phi=\theta+\beta-\alpha$.
In the second step, we use {\bf Common-Phase Invariance}.

\paragraph{Subadditivity.}
By Lemma~\ref{lem:projective-minimum-exists}, there exist
$\alpha,\beta\in\mathbb R$ such that
\(
\widehat d(U_1,V_1)
&=&
d(U_1,e^{i\alpha}V_1)
\\
\widehat d(U_2,V_2)
&=&
d(U_2,e^{i\beta}V_2).
\)
Therefore,
\[
\begin{aligned}
\widehat d(U_1U_2,V_1V_2)
&\leq
d(U_1U_2,e^{i(\alpha+\beta)}V_1V_2)
\\
&=
d\bigl(
U_1U_2,
(e^{i\alpha}V_1)(e^{i\beta}V_2)
\bigr)
\\
&\leq
d(U_1,e^{i\alpha}V_1)
+
d(U_2,e^{i\beta}V_2)
\\
&=
\widehat d(U_1,V_1)
+
\widehat d(U_2,V_2).
\end{aligned}
\]

\paragraph{Continuity.}
Let $\epsilon>0$, let $t_0\in\mathbb R$, and let $A$ be Hermitian.
By {\bf Continuity}, there exists $\delta>0$ such that, for every
$t\in\mathbb R$,
\[
|t_0-t|<\delta
\quad\Longrightarrow\quad
d(e^{it_0\pi A},e^{it\pi A})<\epsilon.
\]
For such a $t$, choosing $\theta=0$ in the definition of
$\widehat d$ gives
\[
\widehat d(e^{it_0\pi A},e^{it\pi A})
\leq
d(e^{it_0\pi A},e^{(i\cdot 0)} e^{it\pi A})
=
d(e^{it_0\pi A},e^{it\pi A})
<
\epsilon.
\]
Thus $\widehat d$ satisfies {\bf Continuity}.

In summary, $\widehat d$ is a projective distance measure.
\end{proof}

Notice that the proof of Theorem~\ref{lem:from-traditional-distance-measure-to-projective-distance-measure} uses Lemma~\ref{lem:projective-minimum-exists} in a single place: in the choice of minimizing phases in the proof of {\bf Subadditivity}.


\section{Decomposition Lemmas}
\label{sec:decomposition-lemmas}

\subsection{Decomposition of Single-Qubit Gates}

\begin{lemma}[Euler decomposition]
\label{lem:zxz-decomposition}
For every $U\in SU(2)$, there exist
$\alpha,\beta,\gamma\in\mathbb R$ such that
\[
U
=
R_z(\alpha)\,
R_x(\beta)\,
R_z(\gamma).
\]
\end{lemma}

\begin{proof}
By the Euler decomposition for $SU(2)$
from~\citeN{NielsenChuang00}, Theorem 4.1, there exist
$\alpha,\beta,\gamma\in\mathbb R$ such that
\(
U &=& R_z(\alpha)\ R_y(\beta)\ R_z(\gamma).
\)
By {Equations~\eqref{eq:rx-in-terms-of-rz} and
\eqref{eq:ry-in-terms-of-rz}},
\(
R_y(\beta) &=& S\ R_x(\beta)\ S^\dagger.
\)
Moreover,
\[
S=e^{i\pi/4}\ R_z(\pi/2)
\quad\text{and}\quad
S^\dagger=e^{-i\pi/4}\ R_z(-\pi/2).
\]
Therefore, the two global phases cancel, and
\(
U
&=& R_z(\alpha)\ S\ R_x(\beta)\ S^\dagger\ R_z(\gamma)
\;\;=\;\;
  R_z(\alpha+\pi/2)\
  R_x(\beta)\
  R_z(\gamma-\pi/2).
\)
This is the required decomposition.
\end{proof}

\begin{lemma}[Generalized Euler decomposition]
\label{lem:generalized-euler-decomposition}
Let $A_1,A_2\in\mathbb{C}^{2\times 2}$ be Hermitian matrices such that $A_1^2=A_2^2=I$ and $A_1A_2=-A_2A_1$.
For every $U\in SU(2)$, there exist $a,b,c\in\mathbb{R}$ such that
\[
U=e^{iaA_1}e^{ibA_2}e^{icA_1}.
\]
\end{lemma}

\begin{proof}
From $A_1^2=I$, every eigenvalue $\lambda$ of $A_1$ satisfies
$\lambda^2=1$, and hence $\lambda\in\{1,-1\}$.
Since $A_1A_2=-A_2A_1$ and $A_2$ is invertible, we have
$A_1\neq\pm I$. Hence the eigenvalues of $A_1$ are $\{1,-1\}$.

Choose a normalized eigenvector $v_+$ of $A_1$ associated with
eigenvalue $1$, and let
\[
v_-=A_2v_+.
\]
Since $A_2$ is Hermitian and satisfies $A_2^2=I$, it is unitary,
so $v_-$ is normalized. Moreover,
\[
A_1v_-
=
A_1A_2v_+
=
-A_2A_1v_+
=
-A_2v_+
=
-v_-.
\]
Thus $v_-$ is an eigenvector of $A_1$ associated with eigenvalue
$-1$, and $\{v_+,v_-\}$ is an orthonormal basis.

Let $V$ be the matrix whose columns are $v_+$ and $v_-$.
Since $\{v_+,v_-\}$ is an orthonormal basis, $V$ is unitary.
Furthermore,
\[
V^\dagger A_1V=Z,
\qquad
V^\dagger A_2V=X.
\]

Since $U\in SU(2)$, we have $V^\dagger UV\in SU(2)$.
Thus, by Lemma~\ref{lem:zxz-decomposition}, there exist
$\alpha,\beta,\gamma\in\mathbb R$ such that
\[
V^\dagger UV
=
R_z(\alpha)R_x(\beta)R_z(\gamma)
=
e^{-i\alpha Z/2}
e^{-i\beta X/2}
e^{-i\gamma Z/2}.
\]
Conjugating by $V$ and using
$Ve^MV^\dagger=e^{VMV^\dagger}$ yields
\[
U
=
e^{-i\alpha A_1/2}
e^{-i\beta A_2/2}
e^{-i\gamma A_1/2}.
\]
The result follows by taking
$a=-\alpha/2$, $b=-\beta/2$, and $c=-\gamma/2$.
\end{proof}

\subsection{Decomposition of Multi-Qubit Gates}

\begin{lemma}[Cosine-Sine decomposition]
For one n-qubit gate $U$, there exist $2^{n-1}$ 1-qubit gates $R_y(\theta_0)$, $R_y(\theta_1)$, $\cdots$, and $R_y(\theta_{2^{n-1}-1})$, four $(n-1)$-qubit gates $P_0$, $P_1$, $Q_0$, and $Q_1$, and three n-qubit gates $P = \outer{0}{0} \otimes P_{0} + \outer{1}{1} \otimes P_{1}$, $R = \sum^{2^{n-1}-1}_{i=0}R_y(\theta_i) \otimes \outer{i}{i}$, and $Q= \outer{0}{0} \otimes Q_{0} + \outer{1}{1} \otimes Q_{1}$, 
such that
\[ 
\Qcircuit @C=0.5em @R=0.4em {
  & \qw & \multigate{2}{U} & \qw & &&  & & \qw & \gate{} \qwx[2]& \gate{R}      &\gate{} \qwx[2]&\qw\\
 &  &                &     & &=& &   &    &          &                 &               &   \\
 & \qw / & \ghost{U}        & \qw &  && & & \qw / & \gate{Q}       & \gate{} \qwx[-2]&\gate{P}       &\qw\\
}
\]
\label{cosinesine}
\end{lemma}

\begin{proof}
This is the Cosine--Sine decomposition of Paige and
Wei~\cite{PaigeWei94}, applied to the $2^{n-1}\times 2^{n-1}$
block partition of $U$.  In this case, the decomposition has the form
\[
U
=
\begin{pmatrix}
P_0 & 0\\
0   & P_1
\end{pmatrix}
\begin{pmatrix}
C & -S\\
S & C
\end{pmatrix}
\begin{pmatrix}
Q_0 & 0\\
0   & Q_1
\end{pmatrix},
\]
where $C$ and $S$ are diagonal real matrices satisfying
$C^2+S^2=I^{\otimes (n-1)}$.  Hence there exist angles
$\theta_0,\ldots,\theta_{2^{n-1}-1}$ such that the middle matrix is
the multiplexed rotation
\[
R=\sum_{i=0}^{2^{n-1}-1}
R_y(\theta_i)\otimes\outer{i}{i}.
\]
The block-diagonal matrices are precisely the gates $P$ and $Q$
defined above.  This quantum-circuit formulation is also given
by~\citeN{ShendeBullockMarkov06}.
\end{proof}

\begin{lemma}[Quantum multiplexor decomposition] \cite[Theorem 12]{ShendeBullockMarkov06}.
For two $(n-1)$-qubit gates $U_{0}$ and $U_{1}$ and one n-qubit gate $U = \outer{0}{0} \otimes U_{0} + \outer{1}{1} \otimes U_{1}$, there exist $2^{n-1}$ 1-qubit gates $R_z(\alpha_0)$, $R_z(\alpha_1)$, $\cdots$, and $R_z(\alpha_{2^{n-1}-1})$, two $(n-1)$-qubit gates $P$ and $Q$, and one n-qubit gate $R = \sum^{2^{n-1}-1}_{i=0}R_{z}(\alpha_i) \otimes \outer{i}{i}$, such that:
\[ 
\Qcircuit @C=0.5em @R=0.4em {
& \qw & \gate{} \qwx[2] & \qw &  && &  & \qw & \qw       & \gate{R}      &\qw            &\qw\\  
& & &     & & = & &  &             &        &         &               &   \\
& \qw / & \gate{U}       & \qw &  && &  & \qw / & \gate{Q}       & \gate{} \qwx[-2]&\gate{P}       &\qw\\
}
\]
\label{demultiplexing}
\end{lemma}


\begin{lemma} [Quantum multiplexor decomposition] 
\cite[Theorem~8]{ShendeBullockMarkov06}.
For $n \geq 1$,
let $R$ be the $(n+1)$-qubit gate 
$\sum_{i=0}^{2^n-1} R_z(\alpha_i) \otimes \outer{i}{i}$.
We can pick angles $\beta_i$ and $\gamma_i$ (for $0 \leq i < 2^{n-1}$) and
define 
$R_0 = \sum^{2^{n-1}-1}_{i=0}R_{z}(\beta_i) \otimes \outer{i}{i}$ and 
$R_1 = \sum^{2^{n-1}-1}_{i=0}R_{z}(\gamma_i) \otimes \outer{i}{i}$ 
such that:
\[ 
\Qcircuit @C=0.5em @R=0.4em {
  & \qw & \gate{R} \qwx[1] & \qw & &&  & & \qw & \targ & \gate{R_0}      & \targ & \gate{R_1} & \qw\\
 &  \qw &   \gate{} \qwx[1]     &  \qw   & &=& &   &   \qw &     \ctrl{-1}     &  \qw               &        \ctrl{-1}       & \qw  & \qw \\
 & \qw / & \gate{}        & \qw &  && & & \qw / & \qw       & \gate{} \qwx[-2] & \qw  &\gate{} \qwx[-2]    &\qw\\
}
\]
\label{rzrz}
\end{lemma}


\bibliographystyle{ACM-Reference-Format}
\bibliography{references}

\end{document}